\documentclass[submission,copyright,creativecommons]{eptcs}
\providecommand{\event}{LSFA 2022}
\usepackage{graphicx}
\usepackage{stmaryrd}
\usepackage[utf8]{inputenc}
\usepackage{mathtools, amsmath, amssymb, tikz, 
tasks, stackengine, graphicx,enumitem} 
\usetikzlibrary{arrows,automata}
\usetikzlibrary{positioning}
\usepackage{slashed}
\usepackage{trimclip}

\newlist{mylist}{enumerate*}{1}
\setlist[mylist]{label=(\roman*)}

\def\@clipped@vdash{%
  \raise .6ex\hbox{\clipbox{0pt .6ex 0pt .6ex}{$\vdash$}}%
}
\newcommand*\vDdashA{%
  \mathrel{%
    \ooalign{%
      $\vdash$\cr
      \raise  .3ex\hbox{\@clipped@vdash}\cr
      \raise -.3ex\hbox{\@clipped@vdash}%
    }%
  }%
}

\DeclareMathOperator*{\Max}{Max}

\DeclareMathOperator*{\Min}{Min}

\newcommand{\Pcat}{\mathsf{Pt}}
\newcommand{\setcat}{\mathsf{Set}}

\usepackage{amsthm}
\newtheorem{definition}{Definition}
\newtheorem{example}{Example}

\newtheorem{lemma}{Lemma}

 \newcommand{\subseteqw}{\sqsubseteq}

\title{Paraconsistent Transition Systems\thanks{This work is supported by by FCT, the Portuguese funding agency for Science and Technology with the projects UIDB/04106/2020 and PTDC/CCI-COM/4280/2021}}
\author{Ana Cruz
\institute{INESC TEC, University of Minho, Portugal}
\and
Alexandre Madeira 
\institute{CIDMA, University of Aveiro, Portugal}
\and
Luís S. Barbosa
\institute{INESC TEC, University of Minho, Portugal}
}
\def\titlerunning{Paraconsistent Transition Systems}
\def\authorrunning{Cruz, Madeira \& Barbosa}
\begin{document}
\maketitle
\begin{abstract}
Often in Software Engineering a modelling formalism has to support scenarios of inconsistency in which several requirements either reinforce or contradict each other. Paraconsistent transition systems are proposed in this paper as  one such formalism: states evolve through two  accessibility relations capturing weighted evidence of a transition or its absence, respectively.  Their weights come from
a specific residuated lattice. A category of these systems, and the corresponding algebra, is defined providing a formal setting to model different application scenarios. One of them, dealing with the effect of quantum decoherence in quantum programs, is used for illustration purposes.

\end{abstract}

\section{Introduction} \label{sc:intro}

Dealing with application scenarios where requirements either reinforce or contradict each other is not uncommon in  Software Engineering. One such scenarios comes from current practice in quantum computation in the context of  NISQ (\emph{Noisy Intermediate-Scale Quantum}) technology \cite{Preskill18} in which  levels of decoherence of quantum memory need to be articulated with the length of the circuits to assess program quality.

In a recent paper \cite{CMB22}, the authors introduced a new kind of weighted transitions systems which records, for each transition, a positive and negative weight which, informally, capture the degree of effectiveness (\emph{`presence'}) and of impossibility (\emph{`absence'}) of a transition. This allows the model to capture both \emph{vagueness}, whenever both weights sum less than 1, as usual e.g. in fuzzy systems, and \emph{inconsitency}, when their sum exceeds 1.
This last feature motivates the qualifier \emph{paraconsistent} borrowed from the work on paraconsistent logic \cite{js123,CKB07}, which
accommodates inconsistency in a controlled way, treating inconsistent information as potentially informative. Such logics were 
originally developed in Latin America in the decades of 1950 and 1960, mainly by F. Asenjo and Newton da Costa. Quickly, however,
the topic attracted attention in the international community  and the original scope of mathematical applications broadened  out, as witnessed in a recent book  emphasizing the engineering potential of paraconsistency \cite{Aka16}. In particular, a number of applications to themes from quantum mechanics and quantum information theory have been studied by D. Chiara \cite{CG00} and W. Carnielli and his collaborators \cite{AC10,CK14}. 

This paper continues such a research program in two directions. First it introduces a suitable notion of morphism for paraconsistent labelled transition systems (PLTS)  leading to the definition of the corresponding category and its algebra. Notions of simulation, bisimulation and trace for PLTS are also discussed. On a second direction, the paper discusses an application of PLTS to reason about the effect of quantum decoherence in quantum programs.

\paragraph{Paper structure.} After recalling the concept of a PLTS and defining their morphisms in  section \ref{sc:def}, section \ref{sc:mor} discusses suitable notions of simulation, bisimulation and trace. Compositional construction of (pointed) PLTS are characterised in section \ref{sc:new} by exploring the relevant category, following G. Winskel  and M. Nielsen's `recipe' \cite{WN95}. Section \ref{sc:qu} illustrates their use to express quantum circuits with decoherence. Finally, section \ref{sc:conc} concludes and points out a number of future research directions.

\section{Paraconsistent labelled transition systems}\label{sc:def}
A \emph{paraconsistent labelled transition system} (PLTS) incorporates two  accessibility relations, classified as positive and negative, respectively, which characterise each transition in opposite ways: one represents the evidence of its presence and other the evidence of its absence.  Both relations are weighted by elements of a residuated lattice
 $\Sigma = \langle \wedge,\vee,\odot ,\to,\text{1},\text{0}\rangle$, where,
$\langle A, \wedge,\vee,\text{1},\text{0}\rangle$ is a lattice,
$\langle A,\odot, 1\rangle$ is a monoid, 
and operation $\odot$ is residuated, with $\to$, i.e. for all $a,b,c\in A$,
 $a\odot b\leq c \Leftrightarrow b\leq a\to c$.
 A G\"odel algebra $G=\langle [0,1], min, max, min, \to,0,1\rangle$ is an example of such a structure, that will be used in the sequel. Operators $max$ and $min$ retain the usual definitions, whereas implication is given by
 $$
 a \to b = \begin{cases} 1, \text{ if } a \leq b \\ b, \text{ otherwise } \end{cases} .
 $$
 Our constructions, however, are, to a large extent, independent of the particular residuated lattice chosen. The definition below extends the one in reference \cite{CMB22} to consider labels in an explicit way. Thus,

\begin{definition}\label{plts} A \textbf{paraconsistent labelled transition system} (PLTS)  over a residuted lattice $A$ and a set of atomic actions $\Pi$ is a structure $\langle W,R,\Pi \rangle$ where,
$W$ is a non-empty set of states, $\Pi$ is a set of labels, and 
$R \subseteq W\times\Pi\times W\times A \times A$ characterises its dynamics, subjected to the following condition: 
between two arbitrary states there is at most one transition involving label $a$,
for every $a \in \Pi$. Each tuple $(w_1,a,w_2,\alpha,\beta)\in R$ represents a transition from $w_1$ to $w_2$ labelled by $(a,\alpha, \beta)$, where $\alpha$ is the degree to which the action $a$ contributes to a transition from $w_1$ to $w_2$, and $\beta$, dually, expresses the degree to which it prevents its occurrence.
\end{definition}

The condition imposed in the definition above makes it possible to express relation $R$ in terms of a \emph{positive} and a \emph{negative} accessibility relation
$r^+, r^- :\Pi \longrightarrow A^{W \times W}$, with
$$r^+(\pi)(w, w')=\begin{cases}  \alpha \text{ if }(w,\pi,w',\alpha,\beta)\in R \\ 0 \text{ otherwise } \end{cases}$$
and $r^-$ defined similarly.
These two relations jointly express different behaviours associated to a transition:
\begin{itemize}[noitemsep]
\item \emph{inconsistency}, when the positive and negative weights  are contradictory, i.e. they sum to some value greater then $1$; this corresponds to the upper triangle in the picture below, filled in grey.
\item \emph{vagueness}, when the sum is less than $1$,  corresponding to the lower, periwinkle  triangle in the same picture;
\item \emph{consistency}, when the sum is exactly $1$, which means that the measures of the factors enforcing or preventing a transition are complementary, corresponding to the red line in the picture.
\end{itemize}
\begin{center}
\scalebox{.7}{
\begin{tikzpicture}
\draw[gray, ultra thick] (0,0) rectangle (3,3);
\draw[black!40!green, ultra thick] (-0.3,0) -- (-0.3,3) ;
\draw[blue, ultra thick] (0,-0.3) -- (3, -0.3);
\draw[] (-1,3.3) node[anchor=east, rotate=90] {Transition is present};
\draw[] (1.5,-0.9) node[anchor=north] {Transition is absent};
\filldraw[blue] (0,-.3) circle (2pt) node[anchor=north] {$0$};
\filldraw[blue] (3,-0.3) circle (2pt) node[anchor=north] {$1$};
\filldraw[black!40!green] (-0.3,0) circle (2pt) node[anchor=east] {$0$};
\filldraw[black!40!green] (-0.3,3) circle (2pt) node[anchor=east] {$1$};
\filldraw[gray] (3,0) -- (3,3) -- (0,3);
\filldraw[white!60!blue] (0,0) -- (0,3) -- (3,0);
\draw[red, ultra thick] (0,3) -- (3,0);
\end{tikzpicture}
}
\end{center}

Morphisms between PLTS respect, as one would expect, the structure of both accessibility relations. Formally,

\begin{definition} \label{morph} 
Let $T_1=\langle W_1,R_1, \Pi \rangle$, $T_2=\langle W_2, R_2, \Pi \rangle$ be two PLTSs defined over the same set of actions $\Pi$. 
 A \textbf{morphism} from $T_1$ to $T_2$ is a function $h:W_1\rightarrow W_2$ such that 
 $$
 \forall_{a \in\Pi},\;  {r_1^+}(a)(w_1,w_2)\leq {r_2^+}(a)(hw_1,hw_2)\;  \text{and}\;  {r_1^-}(a)(w_1,w_2)\geq {r_2^-}(a)(hw_1,hw_2)
 $$
\end{definition}

\begin{example} Function $h=\{w_1\mapsto v_1,w_2\mapsto v_2,w_3\mapsto v_3\}$ is a morphism from 
$M_1$ to $M_2$, over $\Pi=\{a,b,c,d\}$, depicted below

\begin{center}
\begin{tabular}{cc}
\begin{tikzpicture}[->]
	\node[state] (w_1) {$w_1$};
	\node[state] (w_2) [below = of w_1] {$w_2$};
	\node[state] (w_3) [right = of w_2] {$w_3$};
	\node[state] (w_4) [below = of w_3] {$w_4$};
	\draw (w_1) edge[left] node {$(a,0.7,0.2)$} (w_2);
	\draw (w_2) edge[bend left, above] node {$(b,0.3,0.5)$} (w_3);
	\draw (w_3) edge[bend left, below] node {$(c,0.2,0.3)$} (w_2);
	\draw (w_3) edge[right] node {$(d,0.5,0.8)$} (w_4);
\end{tikzpicture}
\qquad &
\begin{tikzpicture}[->]
	\node[state] (v_1) {$v_1$};
	\node[state] (v_2) [below = of v_1] {$v_2$};
	\node[state] (v_3) [right = of v_2] {$v_3$};
	\node[state] (v_4) [below = of v_3] {$v_4$};
	\node[state] (v_5) [above = of v_3] {$v_5$};
	\draw (v_1) edge[left] node {$(a,0.9,0.1)$} (v_2);
	\draw (v_2) edge[bend left, above] node {$(b,0.5,0.2)$} (v_3);
	\draw (v_3) edge[bend left, below] node {$(c,0.6,0.1)$} (v_2);
	\draw (v_3) edge[right] node {$(c,0.8,0.4)$} (v_4);
	\draw (v_3) edge[right] node {$(a,0.4,0.7)$} (v_5);
\end{tikzpicture}
\end{tabular}
\end{center}
\end{example}

\section{Simulation and Bisimulation for PLTS}\label{sc:mor}

Clearly, PLTSs and their morphisms form a category, with composition and identities borrowed from $\mathsf{Set}$. To compare PLTSs is also useful to define what simulation and bisimulation mean in this setting. Thus, under the same assumptions on $T_1$ and $T_2$, 

\begin{definition}\label{simulation} 
A relation $S \subseteq W_1 \times W_2$ is a \textbf{simulation} provided that, for all $\langle p,q \rangle \in S$ and $a\in\Pi$,
\begin{equation*}
p\xrightarrow{(a,\alpha,\beta)\text{ }}_{T_1}p' \Rightarrow \langle \exists_{q'\in W_2}. \, \exists_{\gamma,\delta\in[0,1]} .\; q  \xrightarrow{(a,\text{ }\gamma,\text{ } \delta)\text{ }}_{T_2}q' \; \wedge \; \langle p',q'\rangle\in S \; \wedge\;  \gamma \geq \alpha \;\wedge \; \delta \leq \beta\rangle
\end{equation*}
which can be abbreviated to
\begin{equation*}
p\xrightarrow{(a,\alpha,\beta)\text{ }}_{T_1}p' \Rightarrow \langle \exists_{q'\in W_2} . \; q  \xrightarrow{(a,\text{ }\gamma: \text{ } \gamma \geq \alpha\text{ },\text{ } \delta:\text{ }\delta\leq\beta)\text{ }}_{T_2}q' \; \wedge\;  \langle p',q'\rangle\in S\rangle
\end{equation*}
Two states $p$ and $q$ are \textbf{similar}, written $p \lesssim q$, if there is a simulation $S$ such that $\langle p,q \rangle \in S$. 
\end{definition}
Whenever one restricts in the definition above to the existence of  values $\gamma$ (resp. $\delta$) such that $\gamma \geq \alpha$
(resp. $\delta \leq \beta$), the corresponding simulation is called \emph{positive} (resp. \emph{negative}).

\begin{example}\label{simeg} In the PLTSs depicted  below, $w_1\lesssim v_1$, witnessed by  
$$S=\{ \langle w_1,v_1\rangle,\langle w_2,v_2\rangle, \langle w_3,v_2\rangle, \langle w_4,v_3 \rangle, \langle w_5,v_4\rangle\}$$
\begin{center}
\scalebox{.7}{
\begin{tikzpicture}[->]
	\node[state] (w_1) {$w_1$};
	\node[state] (w_2) [right = of w_1] {$w_2$};
	\node[state] (w_3) [below = of w_2] {$w_3$};
	\node[state] (w_4) [right = of w_2] {$w_4$};
	\node[state] (w_5) [right = of w_3] {$w_5$};
	\draw (w_1) edge[above] node[above=0.5] {$(a,0.4,0.7)$} (w_2);
	\draw (w_1) edge[above] node[below,rotate=-45] {$(a,0.3,0.6)$} (w_3);
	\draw (w_2) edge[above] node[above=0.5] {$(b,0.2,0.8)$} (w_4);
	\draw (w_3) edge[above] node[below=0.5] {$(c,0.2,0.9)$} (w_5);
\end{tikzpicture}
\qquad
\begin{tikzpicture}[->]
	\node[state] (v_1) {$v_1$};
	\node[state] (v_2) [right = of v_1] {$v_2$};
	\node[state] (v_3) [right = of v_2] {$v_3$};
	\node[state] (v_4) [below = of v_3] {$v_4$};
	\draw (v_1) edge[above]  node[above=0.5] {$(a,0.5,0.5)$} (v_2);
	\draw (v_2) edge[above] node[above=0.5] {$(b,0.3,0.5)$} (v_3);
	\draw (v_2) edge[above] node[below, rotate=-45] {$(c,0.5,0.5)$} (v_4);
\end{tikzpicture}
}
\end{center}
\end{example}

\noindent
Finally, 
\begin{definition} \label{bisimulation}
A relation $B \subseteq W_1 \times W_2$ is a \textbf{bisimulation} if   for  $\langle p,q\rangle\in B$ and $a \in \Pi$
\begin{align*}
p\xrightarrow{(a,\alpha,\beta)\text{ }}_{M_1}p' & \Rightarrow \langle \exists q'\in W_2:q  \xrightarrow{(a,\alpha,\beta)\text{ }}_{M_2}q' \wedge \langle p',q'\rangle\in B\rangle\\
q\xrightarrow{(a,\alpha,\beta)\text{ }}_{M_2}q'  &\Rightarrow \langle \exists p'\in W_1:p  \xrightarrow{(a,\alpha,\beta)\text{ }}_{M_1}p' \wedge \langle p',q'\rangle\in B\rangle
\end{align*}

Two states $p$ and $q$ are \textbf{bisimilar}, written $p\sim q$, if there is a bisimulation $B$ such that $\langle p,q\rangle \in B$.
\end{definition} 

\begin{example}
Consider the two PLTSs depicted below. Clearly, $w_1 \sim v_1$.
\begin{center}
\scalebox{.7}{
\begin{tikzpicture}[->]
	\node[state] (w_1) {$w_1$};
	\node[state] (w_2) [below = of w_1] {$w_2$};
	\node[state] (w_3) [below= of w_1, right = of w_2] {$w_3$};
	\draw (w_1) edge[left] node[]{$(a,0.5,0.3)$} (w_2);
	\draw (w_1) edge node[rotate=-45, above]{$(a,0.7,0.2)$} (w_3);
	\draw (w_2) edge node[below=0.4]{$(c,0.2,0.3)$} (w_3);
	\draw (w_3) edge[loop right] node{$(c,0.4,0.5)$} (w_3);
	\draw (w_2) edge[loop left] node{$(c,0.4,0.5)$} (w_2);
\end{tikzpicture}
\qquad
\begin{tikzpicture}[->]
	\node[state] (v_1) {$v_1$};
	\node[state] (v_2) [below = of v_1] {$v_2$};
	\draw (v_1) edge node[right=0.3] {$(a,0.7,0.2)$} (v_2);
	\draw (v_2) edge[loop right] node{$(c,0.4,0.5)$} (v_2);
\end{tikzpicture}
}
\end{center}
\end{example}

\begin{lemma}  
Similarity, $\lesssim$, and  bisimilarity, $\sim$, form a preorder and an equivalence relation, respectively.
\end{lemma}
\begin{proof} The proof is similar to one for classical labelled transition systems (details in \cite{Ana21}).
\end{proof}

As usual, a \emph{trace} from a given state $w$ in a PLTS $T$ is simply the sequence $s$ of tuples $(a, \alpha, \beta)$ labelling a path in $T$ starting at $w$. A first projection on such a sequence, i.e. $\pi_1^* (s)$ retrieves the corresponding sequence of labels
that constitutes what may be called an \emph{unweighted trace}. More interesting is the notion of \emph{weighted trace} which appends to the sequence of labels,  the maximum value for the positive accessibility relation and the minimum value for the negative accessibility relation computed along the trace $s$. Formally,

\begin{definition} 
Given a trace $s$ in a PLTS $T$, the corresponding \emph{weighted trace} is defined by
$$
tw(s)\; =\; \langle \pi_1^*, \bigwedge(\pi_2^*), \bigvee(\pi_3^*) \rangle\, (s)
$$
where, $\pi_n$ denotes the $n$ projection in a tuple, $\langle f,g,h \rangle$ is the universal arrow to a Cartesian product, 
$f^*$ is the functorial extension of $f$ to sequences over its domain, and $\bigwedge$ (resp. $\bigvee$) are the distributed version of $\wedge$ (resp. $\vee$) over sequences.
\end{definition}

\begin{definition}\label{subtrace} A weighted trace   $t=\langle [a_1,a_2,...,a_m], \alpha, \beta \rangle$ is a
\textbf{weighted subtrace} of $t'=\langle [b_1,b_2,...,b_n], \gamma, \delta \rangle$  if 
\begin{mylist} 
\item sequence $ [a_1,a_2,...,a_m]$ is a prefix of $[b_1,b_2,...,b_n]$,
\item $\gamma \geq \alpha$ and
\item $\; \delta \leq \beta$.
\end{mylist}
The definition lifts to sets as follows: given two sets $X$ and $Y$ of weighted traces,
$$X \subseteqw Y\; \; \text{iff}\; \; \forall_{t\in X} .  \exists_{ t'\in Y} . \; t \text{ is a weighted subtrace of } t'$$
\end{definition}

\begin{example} Consider again the two PLTSs given in Example \ref{simeg}. The weighted traces from $w_1$ are $\{t_1=\langle [a,b], 0.2,0.8\rangle, t_2=\langle [a,c] ,0.2,0.9\rangle\}$ and the ones from $v_1$ are $\{t_1'= \langle[a,b], 0.5,0.5\rangle, t_2'=\langle[a,c] ,0.5,0.5\rangle\}$. Clearly, $t_1$  (resp. $t_2$) is a weighted  subtrace of $t_1'$ (resp. $t_2'$).
\end{example}

\begin{lemma}\label{simtraces} Consider two PLTSs, $T_1=\langle W_1, R_1 \rangle$ and $T_2= \langle W_2, R_2\rangle$. If two states $p\in W_1$ and $q\in W_2$ are similar (resp. bisimilar), \textit{i.e.}, $p \lesssim q$ (resp. $p \sim q$), then the set of weighted traces from $p$, $X$, and the set of weighted traces from $q$, $Y$, are such that $X\subseteqw Y$ (resp. coincide).
\end{lemma}

\begin{proof}
If $p \lesssim q$  each trace $t$ from $p$ is a prefix of trace  $t'$ from $q$.
Let $[\alpha_1,\alpha_2,...,\alpha_m]$ and $[\beta_1,\beta_2,..., \beta_m]$ be the sequences of positive and negative weights associated to $t$. Similarly, let $[\alpha'_1,\alpha'_2,...,\alpha'_n]$ and $[\beta'_1,\beta'_2,..., \beta'_n]$ be the corresponding sequences for $t'$; of course $m \leq n$. As $(p,q)$ belongs to a simulation, $\alpha'_i \geq \alpha_i$ and
$\beta'_i \leq \beta_i$, for all $i \leq n$. So, $Min[\alpha'_1,\alpha'_2,...,\alpha'_m] \geq Min[\alpha_1,\alpha_2,...,\alpha_m]$ and 
$Max[\alpha'_1,\alpha'_2,...,\alpha'_m] \leq Max[\alpha_1,\alpha_2,...,\alpha_m]$. Note that $Min$ and $Max$ correspond to 
$\bigwedge$ and $\bigvee$ in a G\"odel algebra. 
Thus,
$$\langle t, Min [\alpha_1,\alpha_2,...,\alpha_n], Max[\alpha_1,\alpha_2,...,\alpha_n] \rangle$$ is a weighted subtrace of
$\langle t'|_m, Min[\alpha'_1,\alpha'_2,...,\alpha'_n], Max[\alpha'_1,\alpha'_2,...,\alpha'_n] \rangle$, where 
$t'|_m$ is the subsequence of $t$ with $m$ elements.
The statement for $\sim$ follows similarly.
\end{proof}
Note that the converse of this lemma does not hold, as shown by the following counterexample.

\begin{example} Consider the PLTS depicted below.
\begin{center}
\scalebox{.7}{
\begin{tikzpicture}[->]
	\node[state] (w_1) {$w_1$};
	\node[state] (w_2) [right = of w_1] {$w_2$};
	\node[state] (w_3) [right= of w_2] {$w_3$};
	\draw (w_1) edge[bend left, above] node{$(a,0.5,0.3)$} (w_2);
	\draw (w_2) edge[bend left, above] node{$(b,0.7,0.2)$} (w_3);
\end{tikzpicture}
\qquad
\begin{tikzpicture}[->]
	\node[state] (v_1) {$v_1$};
	\node[state] (v_2) [right = of v_1] {$v_2$};
	\node[state] (v_3) [right= of v_2] {$v_3$};
	\draw (v_1) edge[bend left, above] node{$(a,0.7,0.2)$} (v_2);
	\draw (v_2) edge[bend left, above] node{$(b,0.5,0.3)$} (v_3);
\end{tikzpicture}
}
\end{center}
 $X = \{\langle [a], 0.5, 0.3 \rangle, \langle [a,b], 0.5, 0.3 \rangle\}$  is the set of weighted traces from $w_1$. Similarly,\\ 
 $Y =  \{\langle [a], 0.7, 0.2 \rangle, \langle [a,b], 0.5, 0.3 \rangle\}$ is the corresponding set from $w_2$.
Clearly $\langle [a], 0.5, 0.3 \rangle$  is a weighted subtrace of $\langle [a], 0.7, 0.2 \rangle$. 
Thus $X \subseteqw Y$. However, $w_1 \not\lesssim w_2$.

\end{example}

\section{New PLTS from old}\label{sc:new}

New PLTS can be built compositionally. This section introduces the relevant operators by exploring the structure of the category of 
$\Pcat$ of \emph{pointed} PLTS, i.e. whose objects are PLTSs with  a distinguished initial state, i.e. 
$\langle W, i, R, \Pi \rangle$, where $\langle W, R, \Pi \rangle$ is a PLTS and $i \in W$. Arrows in $\Pcat$ are allowed between PLTSs with different sets of labels, therefore generalizing Definition \ref{morph} as follows:

\begin{definition} Let $T_1=\langle W_1, i_1, R_1, \Pi\rangle$ and $T_2 =\langle W_2,i_2,R_2,\Pi '\rangle$ be two pointed PLTSs. A morphism in  $\Pcat$ from $T_1$ to $T_2$ is a pair of functions ($\sigma : W_1\to W_2$, $\lambda : \Pi\to_{\bot}\Pi'$) such that\footnote{Notation $\lambda:\Pi\to_{\bot}\Pi' $ stands for the totalization of a partial function by mapping to $\bot$ all elements of $\Pi$ for which the function is undefined.}
$\sigma(i_1)=i_2$,
and, if $(w,a,w',\alpha,\beta)\in R_1$ then $(\sigma(w),\lambda(a),\sigma(w'),\alpha',\beta')\in {R_2}^{\bot}$, with $\alpha \leq \alpha'$ and $\beta'\leq\beta$, where, for an accessibility relation $R$,   $R^{\bot}=R \cup \{(w,\bot,w,1,0) \mid w \in W\}$ denotes $R$ enriched with \emph{idle} transitions in each state.
\end{definition}

\noindent
Clearly $\Pcat$ forms a category, with composition inherited from $\setcat$ and $\setcat_{\bot}$, the later standing for the category of sets and partial functions, with $T_{\textit{nil}}=\langle \{*\},*,\emptyset,\emptyset \rangle$ as both the initial and final object. The corresponding unique morphisms are $!:T\to T_\text{nil}$, given by $\langle \underline{*}, ()\rangle$, and $?:T_{\textit{nil}}\to T$,
given by $\langle \underline{i}, ()\rangle$, where $()$ is the empty map and notation $\underline{x}$ stands for the constant, everywhere $x$, function.

An algebra of PLTS typically includes some form of parallel composition, disjoint union, restriction, relabelling and prefixing, as one is used from the process algebra literature \cite{BBR10}. Accordingly, these operators are defined along the lines proposed by G. Winskel and M. Mielsen \cite{WN95}, for the standard, more usual case.

%
%
\paragraph{Restriction.}
The restriction operator is intended to control the interface of a transition system, preserving, in the case of a PLTS, the corresponding positive and negative weights. Formally,

\begin{definition} Let $T=\langle W,i,R,\Pi\rangle$ be a PLTS, and  $\lambda : \Pi' \to \Pi$ be an inclusion.
The \textbf{restriction}  of $T$ to $\lambda$, $T\upharpoonright \lambda$, is a PLTS $\langle W,i,R',\Pi'\rangle$ over $\Pi'$ such that 
$R'=\{(w,\pi,w',\alpha,\beta)\in R \mid \pi\in\Pi'\}$.
\end{definition}

There is a morphism $f=(1_W,\lambda)$ from  $T\upharpoonright \lambda$ to $T$, 
and a functor $P:\Pcat \to \setcat_{\bot}$  which sends a morphism $(\sigma,\lambda):T\to T'$ to the partial function $\lambda:\Pi'\to\Pi$. Clearly, $f$ is the Cartesian lifting of morphism $P(f)=\lambda$ in $\setcat_{\bot}$. Being Cartesian means that for 
any $g:T'\to T$ in $\Pcat$ such that $P(g)=\lambda$ there is a unique morphism $h$ such that $P(h)=1_{\Pi '}$ making the following diagram to commute:
\begin{center}
\begin{tikzpicture}[->]
	\node[] (t1) {$T'$};
	\node[] (t2) [below = of t1] {$T\upharpoonright \lambda$};
	\node[] (t3) [right= of t2] {$T$};
	\draw (t1) edge[left] node{$h$} (t2);
	\draw (t1) edge[above] node{$g$} (t3);
	\draw (t2) edge[below] node{$f$} (t3);
\end{tikzpicture}
\end{center}
Note that, in general, restriction does not preserve reachable states. Often, thus, the result of a restriction is itself restricted to its reachable part.

\paragraph{Relabelling.}
In the same group of \emph{interface-modifier} operators, is \emph{relabelling}, which renames the labels of a PLTS according to 
a total function $\lambda:\Pi\to\Pi'$.

\begin{definition} Let $T=\langle W,i,R,\Pi\rangle$ be a PLTS, and  $\lambda : \Pi' \to \Pi$ be a total function.
The \textbf{relabelling}  of $T$ according to $\lambda$, $T\{\lambda\}$ is the PLTS $\langle W,i,R',\Pi ' \rangle$ where 
$R' = \{(w,\lambda(a),w',\alpha,\beta)\mid (w,a,w',\alpha,\beta)\in R \}$.
\end{definition}

%
%

Dually to the previous case, there is a morphism $f=(1_W,\lambda)$ from  $T$ to  $T\{\lambda\}$ which is the  cocartesian lifting of $\lambda$ ($= P(f)$).

%
%
%

\paragraph{Parallel composition.}

The product of two PLTSs combines their state spaces and includes all \emph{synchronous} transitions, triggered by the simultaneous occurrence of an action of each component, as well as \emph{asynchronous} ones in which a transition in one component is paired with an \emph{idle} transition, labelled by $\bot$, in the other. Formally,

\begin{definition}\label{pltsprod} Let $T_1=\langle W_1,i_1,R_1,\Pi_1\rangle$ and $T_2=\langle W_2,i_2,R_2,\Pi_2\rangle$ be two PLTS. Their \textbf{parallel composition} $T_1\times T_2$ is the PLTS $\langle W_1\times W_2,(i_1,i_2),R,\Pi '\rangle$, such that 
$\Pi'=\Pi_1\times_{\bot}\Pi_2 = \{(a,\bot)\mid a\in\Pi_1\} \cup \{(\bot,b)\mid b\in\Pi_2\}\cup \{(a,b)\mid a \in \Pi_1,b\in\Pi_2\}$, and
 $(w,a,w',\alpha,\beta)\in R$ if and only if $(\pi_1(w),\pi_1(a),\pi_1(w'),\alpha_1,\beta_1) \in {R_1}^{\bot}$, 
$(\pi_2(w),\pi_2(a),\pi_2(w'),\alpha_2, \beta_2) \in  {R_2}^{\bot}$, $\alpha=min(\alpha_1,\alpha_2)$ and $\beta=max(\beta_1,\beta_2)$.
\end{definition}

\begin{lemma}
Parallel composition is the product  construction in $\Pcat$.
\end{lemma}

\begin{proof}
 In the diagram below let $g_i=(\sigma_i,\lambda_i)$, for $i = 1, 2$, and define  $h$ as  $h = (\langle \sigma_1, \sigma_2 \rangle, \langle \lambda_1, \lambda_2 \rangle)$, where
 $\langle f_1, f_2\rangle(x) = (f_1(x), f_2(x))$ is the  universal arrow in a product diagram in $\setcat$. Clearly, $h$ lifts universality to $\Pcat$, as the unique arrow making the diagram to commute. It remains show it is indeed an arrow in the category. 
 Indeed, 
 let $T=\langle W,i,R,\Pi\rangle$, $T_1=\langle W_1,i_1,R_1,\Pi_1\rangle$, and define 
$T_1\times T_2=\langle W_1\times W_2,(i_1,i_2),R',\Pi'\rangle$ according to defintion \ref{pltsprod}. 
Thus, for each 
$(w,a,w',\alpha,\beta) \in R$, there is a transition $(\sigma_1(w),\lambda_1(a),\sigma_1(w'),\alpha_1,\beta_1) \in {R_1}^{\bot}$ such that $\alpha\leq\alpha_1$ and $\beta \geq\beta_1$; and also a transition $(\sigma_2(w),\lambda_2(a),\sigma_2(w'),\alpha_2,\beta_2) \in {R_2}^{\bot}$ such that $\alpha\leq\alpha_1$ and $\beta\geq\beta_2$.
Moreover, there is a transition 
$$(\langle\sigma_1,\sigma_2\rangle (w),\langle\lambda_1,\lambda_2\rangle (a), \langle\sigma_1,\sigma_2\rangle (w'), min(\alpha_1,\alpha_2),max(\beta_1,\beta_2))\in R'$$
 Thus, there is a transition $(\langle\sigma_1,\sigma_2\rangle (w),\langle\lambda_1,\lambda_2\rangle (a), \langle\sigma_1,\sigma_2\rangle (w'),\alpha',\beta'))\in R'$,  for any $(w,a,w',\alpha,\beta) \in R$ ,  such that $\alpha\leq\alpha'$ and $\beta\geq\beta'$. Furthermore, $\langle \sigma_1,\sigma_2\rangle (i) = (\sigma_1(i),\sigma_2(i))=(i_1,i_2)$. This establishes $h$ as a $\Pcat$ morphism.
 \begin{center}
\begin{tikzpicture}[->]
	\node[] (T1) {$T_1$};
	\node[] (T2) [right = of T1] {$T_1\times T_2$};
	\node[] (T3) [right = of T2] {$T_2$};
	\node[] (T4) [below = of T2] {$T$};
	\draw (T2) edge[above] node{$\Pi_1$} (T1);
	\draw (T2) edge[above] node{$\Pi_2$} (T3);
	\draw (T4) edge[right] node{$h$} (T2);
	\draw (T4) edge[below] node{$g_1$} (T1);
	\draw (T4) edge[below] node{$g_2$} (T3);
\end{tikzpicture}
\end{center}
\end{proof}

\begin{example}\label{egprodplts} Consider the two PLTSs, $T_1$ and $T_2$, depicted below.
\begin{center}
\scalebox{.7}{
\begin{tikzpicture}[->]
	\node[state] (i_1) {$i_1$};
	\node[state] (w) [right = of w_1] {$w$};
	\draw (i_1) edge[bend left, above] node{$(a,0.7,0.2)$} (w);
\end{tikzpicture}
\qquad
\begin{tikzpicture}[->]
	\node[state] (i_2) {$i_2$};
	\node[state] (v) [right = of v_1] {$v$};
	\draw (i_2) edge[bend left, above] node{$(b,0.4,0.2)$} (v);
\end{tikzpicture}
}
\end{center}
Their product $T$ is the PLTS 
\begin{center}
\scalebox{.7}{
\begin{tikzpicture}[->]
	\node[] (i_2) {$(i_1,i_2)$};
	\node[] (v) [right = 3.5cm of v_1] {$(w,i_2)$};
	\node[] (v2) [below = 3.5cm of v] {$(w,v)$};
	\node[] (v3) [below = 3.5cm of i_2] {$(i_1,v)$};
	\draw (i_2) edge[above] node{$((a,\bot),0.7,0.2)$} (v);
	\draw (v) edge[right] node{$((\bot,b),0.4,0.6)$} (v2);
	\draw (i_2) edge[above] node[rotate=-43]{$((a,b),0.4,0.2)$} (v2);
	\draw (i_2) edge[left] node{$((\bot,b),0.4,0.6)$} (v3);	
	\draw (v3) edge[below] node{$((a,\bot),0.7,0.2)$} (v2);
\end{tikzpicture}
}
\end{center}
\end{example}

A suitable combination of parallel composition and restriction may enforce different synchronization disciplines. 
For example, \emph{interleaving} or \emph{asynchronous product} $T_1\interleave T_2$ is defined as
$(T_1\times T_2) \upharpoonright \lambda$
 with the inclusion $\lambda:\Pi\to\Pi_1\times_{\bot}\Pi_2$ for 
$\Pi=\{(a,\bot)\mid a\in\Pi_1\}\cup\{(\bot,b)\mid b\in \Pi_2\}$. This
results in a PLTS $\langle W_1\times W_2,(i_1,i_2),R,\Pi\rangle$ such that 
 $R=\{(w,a,w',\alpha,\beta)\in R'\mid a\in \Pi\}$.
 
Similarly, the \emph{synchronous product} $T_1\otimes T_2$ is also defined as $(T_1\times T_2) \upharpoonright \lambda$, 
taking now $\Pi=\{(a,b)\mid a\in\Pi_1 \text{ and } b \in \Pi_2\}$ as the domain of $\lambda$.

\begin{example} 
Interleaving and synchronous product of $T_1$ and $T_2$ as in Example \ref{egprodplts}, are depicted below.
\begin{center}
\begin{tabular}{cc}
 \scalebox{.7}{
\begin{tikzpicture}[->]
	\node[] (i_2) {$(i_1,i_2)$};
	\node[] (v) [right = 3.5cm of v_1] {$(w,i_2)$};
	\node[] (v2) [below = 3.5cm of v] {$(w,v)$};
	\node[] (v3) [below = 3.5cm of i_2] {$(i_1,v)$};
	\draw (i_2) edge[above] node{$((a,\bot),0.7,0.2)$} (v);
	\draw (v) edge[right] node{$((\bot,b),0.4,0.6)$} (v2);
	\draw (i_2) edge[left] node{$((\bot,b),0.4,0.6)$} (v3);	
	\draw (v3) edge[below] node{$((a,\bot),0.7,0.2)$} (v2);
\end{tikzpicture}
}
&
\scalebox{.7}{
\begin{tikzpicture}[->]
	\node[] (i_2) {$(i_1,i_2)$};
	\node[] (v) [right = 3.5cm of v_1] {$(w,i_2)$};
	\node[] (v2) [below = 3.5cm of v] {$(w,v)$};
	\node[] (v3) [below = 3.5cm of i_2] {$(i_1,v)$};
	\draw (i_2) edge[above] node[rotate=-43]{$((a,b),0.4,0.2)$} (v2);
\end{tikzpicture}
}\\
$T_1\interleave T_2$  & $T_1\otimes T_2$
\end{tabular}
\end{center}
\end{example}

\paragraph{Sum.}
The sum of two PLTSs corresponds to their non-determinisitic composition: the resulting PLTS behaves as either of its components. Formally, 

\begin{definition}\label{pltssum} Let $T_1=\langle W_1,i_1, R_1 ,\Pi_1\rangle$ and $T_2 = \langle W_2, i_2, R_2,\Pi_2\rangle $ be two PLTSs. Their sum $T_1+T_2$ is the PLTS $\langle W,(i_1,i_2),R,\Pi_1\cup \Pi_2 \rangle$, where
\begin{itemize}[noitemsep]
\item [--] $W=(W_1\times \{i_2\}) \cup ( \{i_1\}\times W_2)$ ,
\item [--] $t\in R$ if and only if there exists a transition $(w,a,w',\alpha,\beta) \in R_1$ such that $t=(\iota_1(w),a,\iota_1(w'),\alpha,\beta)$, or  a transition $(w,a,w',\alpha,\beta)\in R_2$ such that $t = (\iota_2(w),a,\iota_2(w'),\alpha,\beta)$
\end{itemize}
where $\iota_1$ and $\iota_2$ are the left and right injections associated to a coproduct in $\setcat$, respectively.
\end{definition}
 Sum is actually a coproduct in $\Pcat$ (the proof follows the argument used for the product case), making $T_1+T_2$ dual to $T_1 \times T_2$.

\begin{example}\label{egprodplts} 
The sum $T_1 + T_2$, for $T_1, T_2$ defined as in Example \ref{egprodplts}  is given by
\begin{center}
\scalebox{.7}{
\begin{tikzpicture}[->]
	\node[] (i_2) {$(i_1,i_2)$};
	\node[] (v) [right = 2cm of i_2] {$(w,i_2)$};
	\node[] (v3) [below = 2cm of i_2] {$(i_1,v)$};
	\draw (i_2) edge[above] node{$(a,0.7,0.2)$} (v);
	\draw (i_2) edge[left] node{$(b,0.4,0.6)$} (v3);	
\end{tikzpicture}
}
\end{center}

\end{example}

\paragraph{Prefixing.}

As a limited form of sequential composition, prefix appends to a pointed PLTS a new initial state and a new transition to the previous initial state, after which the system behaves as the original one. 

\begin{definition} Let $T=\langle W,i,R,\Pi \rangle$ be a PLTS and $w_{new}$ a fresh state identifier not in $W$. Given an action $a$, and $\alpha,\beta \in[0,1]$, the prefix $(a,\alpha,\beta)T$ is defined as $\langle W \cup  \{w_{new}\},w_{new},R',\Pi\cup\{a\}\rangle$ where
$R'= R \cup (w_{new},a, i, \alpha,\beta)$.
\end{definition}

Since it is not required that the prefixing label is distinct from the ones in the original system, prefixing does not extend to a functor in $\Pcat$, as illustrated in the counterexample below.
This is obviously the case for a category of classical labelled transition systems as well. In both cases, however, prefix extens to a functor if the corresponding categories are restricted to action-preserving morphisms, i.e. in which the action component of a morphism is always an inclusion

\begin{example}
Consider two pointed PLTS $T_1$ and $T_2$ 
\begin{center}
\scalebox{.7}{
\begin{tikzpicture}[->]
	\node[state] (i_1) {$i_1$};
	\node[state] (w) [right = of w_1] {$w$};
	\draw (i_1) edge[bend left, above] node{$(a,0.7,0.2)$} (w);
\end{tikzpicture}
}
\qquad
\scalebox{.7}{
\begin{tikzpicture}[->]
	\node[state] (i_2) {$i_2$};
	\node[state] (v) [right = of v_1] {$v$};
	\draw (i_2) edge[bend left, above] node{$(b,0.8,0.1)$} (v);
\end{tikzpicture}
}
\end{center}
connected by a morphism $(\sigma,\lambda):T_1\to T_2$ such that $\sigma(i_1)=i_2$, $\sigma(w)=v$ and $\lambda(a)=b$.
Now consider the prefixes $(a,1,0)T_1$ and $(a,1,0)T_2$ depicted below. 
\begin{center}
\scalebox{.7}{
\begin{tikzpicture}[->]
	\node[state] (i) {$i$};
	\node[state] (i_1) [right =of i] {$i_1$};
	\node[state] (w) [right = of i_1] {$w$};
	\draw (i) edge[bend left,above] node{$(a,1,0)$} (i_1);
	\draw (i_1) edge[bend left, above] node{$(a,0.7,0.2)$} (w);
\end{tikzpicture}
}
\qquad
\scalebox{.7}{
\begin{tikzpicture}[->]
	\node[state] (i) {$i'$};
	\node[state] (i_2) [right =of i]  {$i_2$};
	\node[state] (v) [right = of i_2] {$v$};
	\draw (i) edge[bend left,above] node{$(a,1,0)$} (i_2);
	\draw (i_2) edge[bend left, above] node{$(b,0.8,0.1)$} (v);
\end{tikzpicture}
}
\end{center}
Clearly, a mapping from the actions in $(a,1,0)T_1$ to the actions in $(a,1,0)T_1$ does not exist so neither exists a morphism between the two systems.
\end{example}

\paragraph{Functorial extensions.}
Other useful operations between PLTSs, typically acting on transitions' positive and negative weights, and often restricted to PLTSs over a specific residuated lattice, can be defined functorially in $\Pcat$. An example involving a PLTS defined over a G\"odel algebra is an operation that 
uniformly increases or decreases the value of the positive (or the negative, or both) weight in all transitions. 
Let
$$a\oplus b=\begin{cases} 1 \text{ if } a+b\geq 1\\0\text{ if } a+b\leq 0\\a+b\text{ otherwise } \end{cases}$$
Thus,
\begin{definition}  Let $T=\langle W,i,R,\Pi\rangle$ be a PLTS. Taking $v\in [-1,1]$, the \textbf{positive $v$-approximation} ${T_{\oplus}}^+_v$ is a PLTS $\langle W,i,R',\Pi\rangle$ where 
\begin{center} $R'=\{(w,\pi,w',\alpha \oplus v,\beta) \mid (w,\pi,w',\alpha,\beta)\in R\}$.\end{center}
\end{definition}
\noindent
The definition extends to a functor in $\Pcat$ which is the identity in morphisms. Similar operations can be defined to act on the negative accessibility relation or both.

%
%
%
%
%

Another useful  operation removes all transitions in a pointed PLTS for which the positive accessibility relation is below a certain value and the negative accessibility relation is above a certain value. Formally, 

\begin{definition} Let $T=\langle W,i,R,\Pi\rangle$ be a pointed PLTS, and $p, n\in [0,1]$. The \textbf{purged} PLTS $T_{{p}\uparrow \downarrow{n}} $ is defined as $\langle W,i,R',\Pi\rangle$ where 
$$R'=\{(w,\pi,w',\alpha,\beta) \mid (w,\pi,w',\alpha,\beta)\in R\text{ and } \alpha \geq p \text{ and } \beta\leq m \}$$
\end{definition} 
Clearly, the operation extends to a functor in $\Pcat$, mapping morphisms to themselves.

\section{An application to quantum circuit optimization}\label{sc:qu}

In a quantum circuit \cite{NC10} decoherence consists in decay of a qubit in superposition to its ground state and may be caused by distinct physical phenomena.
A quantum circuit is effective only if gate operations and measurements are performed to superposition states within a limited period of time after their preparation.
In this section pointed PLTS will be used to model circuits incorporating qubit decoherence as an error factor. 
Typically, coherence is specified as an interval corresponding to a worst and a best case.   We employ the two accessibility relations in a PLTS to model both scenarios simultaneously. 

An important observation for the conversion of quantum circuits to PLTS is that quantum circuits always have a sequential execution. Simultaneous operations performed to distinct qubits are combined using the tensor product  $\otimes$ into a single operation to the whole collection of qubits which forms the state of  the circuit. The latter is described by a sequence of executions $e_1,e_2,e_3,...$ where each $e_i$ is the tensor product of the operations performed upon the state at each step. 
The conversion to a PLTS is straightforward, labelling each transition by the tensor of the relevant gates $O_1 \otimes \cdots \otimes O_m$, for $m$ gates involved, but for the computation of 
the positive and negative accessibility relations, $r^+$ and $r^-$.

The weights of a transition corresponding to the application of a gate $O$ acting over $n$ qubits $q_1$ to $q_n$ are given by
$$
v(O) =
\begin{cases} (
1,0) \text{ if qubits $q_1, \cdots q_n$ are in a definite state } \\ 
\left( \Max_i \, f_{\text{max}}(q_i),\Min_i\, f_{\text{min}}(q_i) \right) \text{ otherwise } 
\end{cases}
$$
\noindent
where $f_{\text{max}}(q)=\frac{\tau_{\text{max}}(q)-\tau_{\text{prep}}(q)}{100}$ and $f_{\text{min}}(q)=\frac{\tau_{\text{min}}(q)-\tau_{\text{prep}}(q)}{100}$, $\tau_{\text{max}}(q)$ and $\tau_{\text{min}}(q)$ are the longest and shortest coherence times of $q$, respectively, and $\tau_{\text{prep}}(q)$ is the time from the preparation of $q$'s superposition to the point after the execution of $O$.
The latter are fixed for each type of quantum gate; reference \cite{Zhang_2019} gives experimentally computed values for them as well as for maximum and minimum values for qubit decoherence.
   	 
Consider, now, a transition $t$ labelled by a $O_1\otimes ... \otimes O_m$    	 
   	 Then,
   	 $r^+ = \Max_{i=1}^n \{\pi_1(v(O_i))\}$ and
   	 $r^- = 1 - \Min_{i=1}^n\{\pi_2(v(O_i))\}$.

\begin{example} Consider the following circuits designed with IBM Quantum Composer:

\begin{center}
    \includegraphics[scale=0.6]{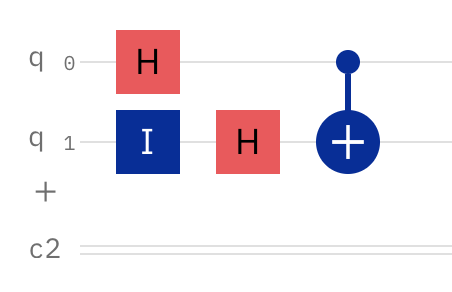}
    \qquad
    \includegraphics[scale=0.6]{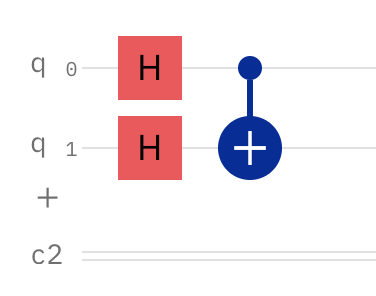}
\end{center}

Assume  that the execution time of a single qubit gate is $\tau_G= 20\mu s$ and of a two qubit gate is $2\tau_G=40 \mu s$ 
\cite{Zhang_2019}, and that both qubits have the same coherence times $\tau_{\text{max}}(q_1)=\tau_{\text{max}}(q_2)=100\mu s$ and $\tau_{\text{min}}(q_1)=\tau_{\text{min}}(q_2)=70\mu s$. Thus the circuit on the left (resp. right) translates into 
$T_1$ (on the left) and  $T_2$  (on the right).

\begin{center}
\begin{tabular}{cc}
\scalebox{.7}{
\begin{tikzpicture}[->]
    \node[] (1) {$q[0]  ; q[1] $};
    \node[] (2) [below = of 1] {$q[0]  ; q[1]$};
    \node[] (3) [below = of 2] {$q[0]  ; q[1]$};
    \node[] (4) [below = of 3] {$q[0]  ; q[1]$};
    \draw (1) edge[right] node{$(H\otimes I,1,0)$} (2);
    \draw (2) edge[right] node{$(I \otimes H,0.8,0.5)$} (3);
    \draw (3) edge[right] node{$(CNOT,0.4,0.9)$} (4);
\end{tikzpicture}
}
\qquad &
\scalebox{.7}{
\begin{tikzpicture}[->]
    \node[] (2) [below = of 1] {$q[0]  ; q[1]$};
    \node[] (3) [below = of 2] {$q[0]  ; q[1]$};
    \node[] (4) [below = of 3] {$q[0]  ; q[1]$};
    \draw (2) edge[right] node{$(H \otimes H,1,0)$} (3);
    \draw (3) edge[right] node{$(CNOT,0.6,0.7)$} (4);
\end{tikzpicture}
}
\end{tabular}
\end{center}

As both  circuits implement the same  quantum algorithm and our focus is only on the effectiveness of the circuits, we 
may abstract from the actual sequences of labels and consider instead
$T_1\{\lambda\}$ and $T_2\{\lambda\}$, for $\lambda$ mapping each label to a unique label 
$\star$.
Their maximal weighted traces \footnote{Such maximal traces are easily identifiable given the peculiar shape of a PLTS corresponding to a quantum circuit.} are 
$$
t_{T_1\{\lambda\}}\, =\,  \langle [*,*,*], 0.4,0.9\rangle\; \; \text{and}\; \; 
t_{T_2\{\lambda\}}\, =\,  \langle [*,*,*], 0.6,0.7\rangle
$$
Clearly  $t_{T_1\{\lambda\}}$ is a weighted subtrace of $t_{T_2\{\lambda\}}$, therefore suggesting a criteria for comparing  the effectiveness of circuits. Indeed, a circuit is more effective (i.e. less affected by qubit decoherence)  than other if the maximal weighted trace of its (relabelled)  PLTS is a weighted subtrace of the corresponding construction in the  other. 

The second circuit is obviously more efficient than the first. This suggests we could use the weighted subtrace relation as a metric to compare circuit quality, for circuits implementing equivalent algorithms. 
\end{example}

Reference \cite{Zhang_2019} introduces a tool which tried to transform a circuit so that  the lifetime of quantum superpositions is shortened. They give several examples of circuits and show how the application of the tool results in a circuit performing the same algorithm but with a reduced error rate. Our next example builds on one of their examples, computes the corresponding  PLTS
and compare the maximal weighted traces.

\begin{example} Consider the following circuits reproduced from  \cite{Zhang_2019},
 which in ideal quantum devices would be indistinguishable.
\begin{center}
\includegraphics[scale=0.7]{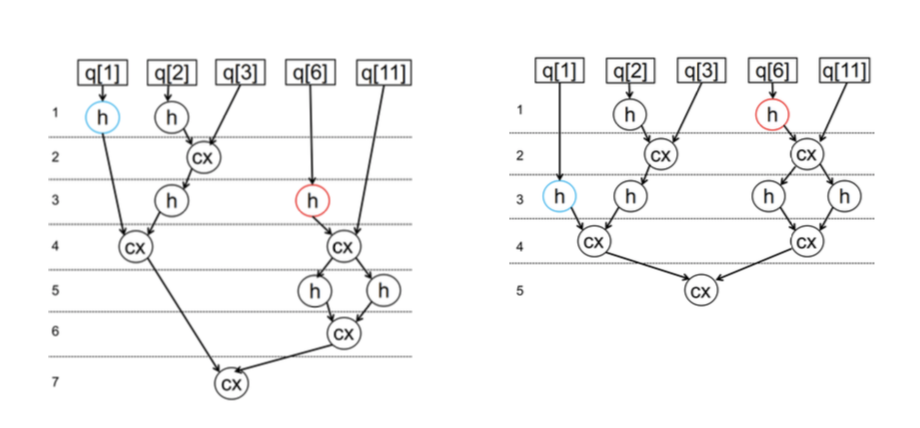}
\end{center}
These circuits are represented as

\begin{center}
\scalebox{.7}{
\begin{tikzpicture}[->]
    \node[] (1) {$s_1$};
    \node[] (2) [below = of 1] {$s_2$};
    \node[] (3) [below = of 2] {$s_3$};
    \node[] (4) [below = of 3] {$s_4$};
    \node[] (5) [below = of 4] {$s_5$};
    \node[] (6) [below = of 5] {$s_6$};
    \node[] (7) [below = of 6] {$s_7$};
    \node[] (8) [below = of 7] {$s_8$};
    \draw (1) edge[left] node{$(H_1 \otimes H_2 ,1,0)$} (2);
    \draw (2) edge[left] node{$(CX_{2,3},0.6,0.7)$} (3);
    \draw (3) edge[left] node{$(H_2 \otimes H_6,0.8,0.5)$} (4);
    \draw (4) edge[left] node{$(CX_{1,2} \otimes CX_{6,11},0,1)$} (5);
    \draw (5) edge[left] node{$( H_6 \otimes H_1,0.8,0.5)$} (6);
    \draw (6) edge[left] node{$(CX_{6,11},0.6,0.7)$} (7);
    \draw (7) edge[left] node{$(CX_{2,6} ,0,1)$} (8);
    
\end{tikzpicture}
\qquad
\begin{tikzpicture}[->]
    \node[] (1) {$r_1$};
    \node[] (2) [below = of 1] {$r_2$};
    \node[] (3) [below = of 2] {$r_3$};
    \node[] (4) [below = of 3] {$r_4$};
    \node[] (5) [below = of 4] {$r_5$};
    \node[] (6) [below = of 5] {$r_6$};
    \draw (1) edge[right] node{$(H_2\otimes H_6 ,1,0)$} (2);
    \draw (2) edge[right] node{$(CX_{2,3}\otimes CX_{6,11} ,0.6,0.7)$} (3);
    \draw (3) edge[right] node{$( H_1 \otimes H_2\otimes H_6 \otimes H_11,0.8,0.5)$} (4);
    \draw (4) edge[right] node{$(CX_{1,2} \otimes CX_{6,11},0.6,0.7)$} (5);
    \draw (5) edge[right] node{$(CX_{2,6},0.6,0.7)$} (6);
\end{tikzpicture}
}
\end{center}
where $H$ and $CX$ are indexed by the numeric identifiers of the qubit(s) to which they apply in each execution step. 
The maximal weighted trace of the (relabelled PLTS corresponding to) circuit in the right,
$\langle [*,*,*,*,*,*,*],0.6,0.7  \rangle$, is a weighted subtrace of the one
corresponding to circuit in the left, 
$\langle [*,*,*,*,*],0,1  \rangle$. Thus, the former circuit is more effective than the latter, as experimentally verified in \cite{Zhang_2019}.
\end{example}

\begin{example} As a final example consider two circuits differing only on the time points in which measurements are placed.  
\begin{center}
\includegraphics[scale=0.6 ]{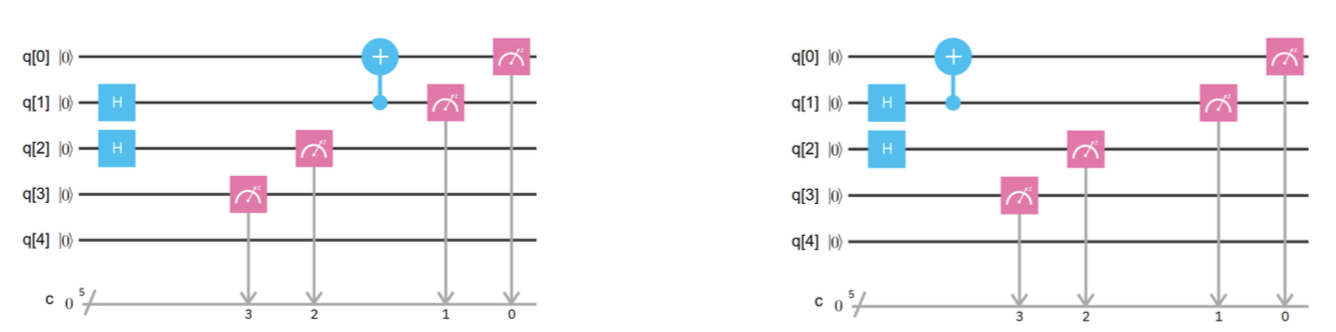}
\end{center}
The corresponding PLTS, computed again with the values given in reference  (where execution  time of a measurement is $\tau_M= 300ns \sim 1\mu s$), are depicted below

\begin{center}
\scalebox{.7}{
\begin{tikzpicture}[->]
    \node[] (1) {$s_1$};
    \node[] (2) [below = of 1] {$s_2$};
    \node[] (3) [below = of 2] {$s_3$};
    \node[] (4) [below = of 3] {$s_4$};
    \node[] (5) [below = of 4] {$s_5$};
    \node[] (6) [below = of 5] {$s_6$};
    \node[] (7) [below = of 6] {$s_7$};
    \draw (1) edge[left] node{$(H_1 \otimes H_2 ,1,0)$} (2);
    \draw (2) edge[left] node{$(M_3,0.99,0.31)$} (3);
    \draw (3) edge[left] node{$(M_2,0.98,0.32)$} (4);
    \draw (4) edge[left] node{$(CX_{0,1},0.58,0.72)$} (5);
    \draw (5) edge[left] node{$(M_1,0.99,0.31)$} (6);
    \draw (6) edge[left] node{$(M_0,0.98,0.32)$} (7);
\end{tikzpicture}
\qquad
\begin{tikzpicture}[->]
    \node[] (1) {$r_1$};
    \node[] (2) [below = of 1] {$r_2$};
    \node[] (3) [below = of 2] {$r_3$};
    \node[] (4) [below = of 3] {$r_4$};
    \node[] (5) [below = of 4] {$r_5$};
    \node[] (6) [below = of 5] {$r_6$};
    \node[] (7) [below = of 6] {$r_7$};
    \draw (1) edge[right] node{$(H_1\otimes H_2 ,1,0)$} (2);
    \draw (2) edge[right] node{$(CX_{0,1} ,0.6,0.7)$} (3);
    \draw (3) edge[right] node{$(M_3,0.99,31)$} (4);
    \draw (4) edge[right] node{$(M_2,0.98,0.32)$} (5);
    \draw (5) edge[right] node{$(M_1,0.97,0.33)$} (6);
    \draw (6) edge[right] node{$(M_0,0.96,0.34)$} (7);
\end{tikzpicture}
}
\end{center}
The maximal  weighted  trace $\langle [*,*,*,*,*,*,],0.6,0.7  \rangle$  corresponding to the circuit on the right is a weighted subtrace of the corresponding one for the circuit on the left, $\langle [*,*,*,*,*,*,], 0.58,0.72  \rangle$. This shows that measuring can be safely 
postponed to the end of a circuit, as experimentally verified.
\end{example}

\section{Conclusions and future work}\label{sc:conc}
The paper introduced a category of a new kind of labelled transition systems able to capture both \emph{vagueness}  and \emph{inconsistency} in software modelling scenarios. The structure of this category was explored to define a number of useful operators to build such systems in a compositional way. Finally, PLTS were used to model effectiveness concerns in the analysis of quantum circuits. 
In this case the  weight corresponding to the `presence' of a transition captures  an index measuring its effectiveness assuming the best case value for qubit decoherence. On the other hand, the  weight corresponding to  the `absence' of a transition measures the possibility of non-occurrence, assuming qubit decoherence worst case value.

A lot remains to be done. First of all, a process logic, as classically associated to labelled transition systems \cite{Sti05}, i.e. a modal logic  with label-indexed modalities, can be designed for pointed PTLS. This will  provide not only yet another behavioural equivalence, based on the set of formulas satisfied by two systems, but also a formal way to express safety and liveness properties of these  systems. 

This will be extremely useful to express and verify properties related to the effectiveness of quantum circuits, therefore pushing further the application scenario proposed in section \ref{sc:qu}.
Finally, automating the construction of a pointed PLTS for a given circuit, parametric on the different  qubit coherence and gate execution time found experimentally, and adding a prover for the logic suggested above, will provide an interesting basis to support quantum circuit optimization.
Reliable, mathematically sound approaches  and tools to   support                                                                         
quantum computer programming  and verification will be part of the quantum research agenda for the years to come. 
Indeed, their lack may put at risk the expected
quantum advantage of the new hardware.  

\bibliographystyle{splncs04}
\bibliography{dissertation}

\end{document}